\documentclass[11pt]{article}

\usepackage{morefloats}
    \usepackage{amsfonts}
\usepackage{color}
\usepackage{multirow}
\usepackage{latexsym}
\usepackage{amsmath,amssymb,amsthm}
\usepackage{dsfont}
\usepackage{extarrows}

    \newcommand{\N}{\mathbb{N}}
    \newcommand{\E}{\mathbb{E}}

    \newcommand{\R}{\mathbb{R}}
    
    \newcommand{\norm}[1]{\left\lVert #1 \right\rVert}

\newtheorem{theorem}{Theorem}[section]
  \newtheorem{lemma}[theorem]{Lemma}

  \theoremstyle{definition}
  
  \newtheorem*{remark}{Remark}
	\renewenvironment{proof}{\textbf{\emph{Proof.}}}{\qed}
	\renewenvironment{remark}{\textbf{\emph{Remark.}}}

\makeatletter
\def\blfootnote{\xdef\@thefnmark{}\@footnotetext}
\makeatother

\begin{document}

\title{\bf A new characterization of the Gamma distribution and associated goodness of fit tests  
}

\author{Steffen Betsch         \and
        Bruno Ebner 
}


\date{\today}
\maketitle

\blfootnote{ {\em American Mathematical Society 2010 subject
classifications.} Primary 62G10 Secondary 60F05}

\begin{abstract}
We propose a class of weighted $L_2$-type tests of fit to the Gamma distribution. Our novel procedure is based on a fixed point property of a new transformation connected to a Steinian characterization of the family of Gamma distributions. We derive the weak limits of the statistic under the null hypothesis and under contiguous alternatives. Further, we establish the global consistency of the tests and apply a parametric bootstrap technique in a Monte Carlo simulation study to show the competitiveness to existing procedures.
\end{abstract}

\section{Introduction}
\label{Intro}
Testing the goodness-of-fit of data to a Gamma distribution is a first step to serious statistical inference involving this model. Due to its versatile nature the Gamma distribution generalizes the exponential, the $\chi^2$ and the Erlang distribution. Applications include modelling rainfall data in Africa, see \cite{HMF:2007}, or honey bee transit times, see \cite{PS:2017}. To be specific, let $X_1,\ldots,X_n,\ldots$ be a sequence of independent and identically distributed (i.i.d.) copies of a positive random variable $X$, all random variables being defined on a common underlying probability space $(\Omega,\mathcal{A},\mathbb{P})$. Writing $\Gamma(k,\lambda)$ for the Gamma distribution with shape parameter $k>0$ and scale parameter $\lambda>0$ as well as $\mathbf{\Gamma}=\{\Gamma(k,\lambda) \, | \, k,\lambda>0\}$ for the family of Gamma distributions, we want to check the assumption that the distribution of $X$ belongs to $\mathbf{\Gamma}$, or equivalently, test the composite hypothesis
\begin{equation*}
	H_0: \mathbb{P}^{X}\in \mathbf{\Gamma}
\end{equation*}
against general alternatives. This testing problem has been considered in the statistical literature. In recent papers, \cite{BG:2015,PS:2017,WM:2008} proposed tests based on some independence properties of the Gamma distribution, \cite{HME:2012} considered a test based on the empirical Laplace transform, \cite{KKC:2006} proposed a method using the empirical moment generating function and \cite{VG:2015} suggested to use a variance ratio test. Most of these tests are built upon a characterization of the Gamma distribution, while classical 'omnibus' methods like the Kolmogorov-Smirnov test, the Cram\'{e}r-von Mises test, the Anderson-Darling test or the Watson test utilize the (weighted) distance of the empirical to the theoretical distribution function in some suitable function space. For a review see \cite{delBarrio2000}. Note that, since $H_0$ stands for testing the fit to a whole family of distributions, every statistic relies on some adequate estimator $(\widehat{k},\widehat{\lambda})$ of $(k,\lambda)$ to choose a 'best' representative of $\mathbf{\Gamma}$.

Our novel idea for this testing problem is to use a fixed point property of a transformation connected to a Stein-type characterization for the Gamma distribution. The family of test statistics is then based on the weighted distance from the empirical transformation to the empirical distribution function. \\

Inspired by the density approach as introduced in \cite{CGS:2011}, section 4, and \cite{LS:2013}, section 2, we first state a characterizing Steinian equation for the Gamma distribution. For $(k, \lambda) \in (0,\infty)^2$ we denote by
\begin{equation*}
	p(t) = p(t,k,\lambda) = \frac{\lambda^{-k}}{\Gamma(k)} \, t^{k-1} \, e^{- t / \lambda}, \quad  t > 0,
\end{equation*}
the density function of the Gamma distribution $\Gamma(k,\lambda)$ with shape parameter $k$ and scale parameter $\lambda$. We write $\mathcal{F}$ for the set of all functions $f : (0, \infty) \to \R$ that are differentiable, satisfy
\begin{equation*}
	\lim_{x \, \searrow \, 0} f(x)
	= \lim_{x \, \searrow \, 0} f(x) p(x)
	= \lim_{x \, \to \, \infty} f(x) p(x)
	= 0
\end{equation*}
and have a locally integrable derivative.
\begin{theorem} \label{thm SC}
	A positive random variable $X$ has a $\Gamma(k,\lambda)$-distribution (with parameters $k,\lambda>0$) if, and only if,
	\begin{equation} \label{density approach}
		\E \left[ f^{\prime}(X) + \left( \frac{k - 1}{X} - \frac{1}{\lambda} \right) f(X) \right] = 0
	\end{equation}
	for every $f \in \mathcal{F}$ for which this expectation exists.
\end{theorem}
For completeness, and as we have changed the class of test functions $\mathcal{F}$ as compared to \cite{LS:2013}, a proof is given in the appendix. \\

The characterization of the Gamma distribution given in Theorem \ref{thm SC} is not directly accessible for the proposal of a goodness-of-fit test, since (\ref{density approach}) depends on the class of functions $\mathcal{F}$. We tackle this problem by the following fixed point property.


\begin{theorem} \label{thm fixed point statement}
	Let $X$ be a positive random variable with distribution function $F$ and assume $\E |X| < \infty$. For $k, \lambda > 0$ define $T^X : \R \to \R$ by
	\begin{align*}
		T^X (t) = \E \left[ \left(- \frac{k - 1}{X} + \frac{1}{\lambda} \right) \min \{ X, t \} \right], \quad  t > 0,
	\end{align*}
	and set $T^X(t) = 0$ for $t \leq 0$. Then $X$ follows the $\Gamma(k, \lambda)$-law if, and only if, $T^X \equiv F$ on $\R$.
\end{theorem}
\begin{proof}
	For the necessity note first that if $X \sim \Gamma(k,\lambda)$, we have
	\begin{align*}
		\E \left[ \left( - \frac{k - 1}{X} + \frac{1}{\lambda} \right) \mathds{1} \{ X > t \} \right]
		&= - \int_0^{\infty} \frac{\partial}{\partial x} \, p(x, k, \lambda) \, \mathds{1} \{ x > t \} \, \mathrm{d} x \\
		&= p(t, k, \lambda) \\
		&\geq 0
	\end{align*}
	for all $t > 0$. Therefore, setting
	\begin{align*}
		d^X(t) = \E \left[ \left( - \frac{k - 1}{X} + \frac{1}{\lambda} \right) \mathds{1} \{X > t\} \right] \, \mathds{1}_{(0,\infty)} (t) , \quad t \in \R,
	\end{align*}
	and realizing that Fubini's Theorem implies
	\begin{align} \label{proof eq 1}
		\int_{\R} d^X (s) \, \mathrm{d}s
		&= \E \left[ \left( - \frac{k - 1}{X} + \frac{1}{\lambda} \right) \int_0^{\infty} \mathds{1} \{ X > s \} \, \mathrm{d}s \right] \nonumber \\
		&= -(k - 1) + \frac{1}{\lambda} \, \E X \nonumber \\
		&= 1,
	\end{align}
	we notice that $d^X$ is the density function corresponding to $T^X$ since
	\begin{align} \label{proof eq 2}
		\int_{-\infty}^{t} d^X (s) \, \mathrm{d}s
		= \E \left[ \left( - \frac{k - 1}{X} + \frac{1}{\lambda} \right) \int_0^{t} \mathds{1} \{ X > s \} \, \mathrm{d}s \right]
		= T^X (t)
	\end{align}
	holds for all $t \in \R$. Now, using the fundamental theorem of calculus and (\ref{density approach}), we get, for any $f \in \mathcal{F}$,
	\begin{align*}
		\int_{\R} f^{\prime}(s) \, d^X(s) \, \mathrm{d}s
		&= \E \left[ \left( - \frac{k - 1}{X} + \frac{1}{\lambda} \right) \int_0^{\infty} f^{\prime}(s) \, \mathds{1} \{ X > s \} \mathrm{d}s \right] \\
		&= \E \left[ - \left( \frac{k - 1}{X} - \frac{1}{\lambda} \right) f(X) \right] \\
		&= \E \big[ f^{\prime}(X) \big]
	\end{align*}
	which implies the claim (note that the set of derivatives of the functions in $\mathcal{F}$ separates probability measures). For the converse, we assume $T^X (t) = F(t)$ for all $t \in \R$, and let $d^X$ be as above. By Tonelli's theorem,
	\begin{align*}
		\int_{\R} \left| d^X (s) \right| \mathrm{d}s
		\leq |k - 1| + \frac{1}{\lambda} \, \E |X|
		< \infty
	\end{align*}
	and thus $\left| d^X \right| < \infty$ $\mathcal{L}^1$-almost everywhere, with $\mathcal{L}^1$ denoting the Lebesgue measure on the Borel sets of the real line. Consequently, Fubini's theorem is applicable and we acknowledge (\ref{proof eq 2}) to hold true. Since $T^X$ is increasing, we have $d^X \geq 0$ $\mathcal{L}^1$-a.e. and as dominated convergence gives us
	\begin{align*}
		\E \left[ \left( - \frac{k - 1}{X} + \frac{1}{\lambda} \right) X \right]
		= \lim\limits_{t \, \to \, \infty} T^X (t)
		= 1,
	\end{align*}
	equation (\ref{proof eq 1}) holds. Therefore, $d^X$ is the density function corresponding to $T^X \equiv F$. Using this, we get
	\begin{align*}
		\E \big[ f^{\prime}(X) \big]
		& = \int_{\R} f^{\prime}(s) \, d^X(s) \, \mathrm{d}s \\
		& = \E \left[ \left( - \frac{k - 1}{X} + \frac{1}{\lambda} \right) \int_0^{\infty} f^{\prime}(s) \, \mathds{1} \{X > s\} \, \mathrm{d}s \right] \\
		& = \E \left[ - \left( \frac{k - 1}{X} - \frac{1}{\lambda} \right) f(X) \right]
	\end{align*}
	for every $f \in \mathcal{F}$. By Theorem \ref{thm SC} we are done.
\end{proof}

\begin{remark}
Note that \cite{BE:2018} used the similar zero bias transformation (introduced by \cite{GR:1997}) for testing normality. By analogy with this transformation, the proof of Theorem \ref{thm fixed point statement} also shows that if $X$ is chosen such that $T^X$ is a distribution function, there exists a random variable $X^{\Gamma} \sim T^X$ satisfying
\begin{align} \label{transform equation}
	\E \big[ f^{\prime}(X^{\Gamma}) \big] = \E \left[ \left( - \frac{k - 1}{X} + \frac{1}{\lambda} \right) f(X) \right]
\end{align}
for each $f \in \mathcal{F}$. Since $T^X$ apparently constitutes the only solution of (\ref{transform equation}), Theorem \ref{thm fixed point statement} assures that the $\Gamma(k, \lambda)$-law is the unique fixed point of the distributional transformation $\mathbb{P}^{X} \mapsto \mathbb{P}^{X^{\Gamma}}$. However, the restrictions on $X$ that render $T^X$ a distribution function are very strict (for instance, $\E X = k \lambda$ is a necessary condition) so, in general, we cannot think of $\mathbb{P}^X \mapsto T^X$ as a distributional transformation. For treatments that focus specifically on distributional transformations related to Stein's method we refer to \cite{GR:2005} and \cite{D:2017}. In the latter, the sign changes of the so called biasing function (which in our case is $B = - p^{\prime} / p$) are taken into account to guarantee that the 'biased distribution' is, indeed, a distribution. The price which is paid for this additional piece of structure is that the explicit representation of the transformation (see Remark 1 (d) in \cite{D:2017}) is substantially more complex, especially when considering that, in the case of the Gamma distribution, the sign change depends on the parameters which, later on, have to be estimated.
\end{remark}

\vspace{4mm}

By Theorem \ref{thm fixed point statement}, we are able to compare $T^X$ to $F$ in order to determine how close a given distribution is to $\Gamma(k, \lambda)$ and construct a goodness-of-fit statistic for testing the hypothesis $H_0$, based on $X_1, \dots, X_n$. In what follows, in addition to $X>0$ $\mathbb{P}$-almost surely, we assume $\E |X| < \infty$. Of course the unknown parameters $k, \lambda$ have to be estimated from the data. In view of the scale invariance of the class of Gamma distributions, we set $Y_{n,j} = X_j / \widehat{\lambda}_n$, $j = 1, \dots, n$, where $\widehat{\lambda}_n = \widehat{\lambda}_n(X_1, \dots, X_n)$ is a consistent, scale equivariant estimator of $\lambda$, i.e., we have
\begin{align*}
	\widehat{\lambda}_n (\beta X_1, \beta X_2, \dots, \beta X_n) = \beta \, \widehat{\lambda}_n (X_1, \dots, X_n), \quad \beta > 0.
\end{align*}
In addition, we let $\widehat{k}_n = \widehat{k}_n (X_1, \dots, X_n)$ be a consistent, scale invariant estimator of $k$ with
\begin{align*}
	\widehat{k}_n (\beta X_1, \beta X_2, \dots, \beta X_n) = \widehat{k}_n (X_1, \dots, X_n), \quad \beta > 0.
\end{align*}
Moreover, we assume that the limit $(\widehat{k}_n, \widehat{\lambda}_n) \stackrel{\mathbb{P}}{\longrightarrow} (k, \lambda) \in (0,\infty)^2$ exists regardless of the underlying distribution of $X$. Here, $\stackrel{\mathbb{P}}{\longrightarrow}$ denotes convergence in probability. Notice that any statistic based on $Y_{n,1}, \dots, Y_{n,n}$ and $\widehat{k}_n$ is scale invariant. Observing that $\widehat{\lambda}_n(Y_{n,1}, \dots, Y_{n,n}) = 1$ and that, under $H_0$, $Y_{n,1}$ approximately follows a $\Gamma(k,1)$-law for some $k \in (0,\infty)$, we propose the Cram\'{e}r-von Mises type test statistic
\begin{align} \label{teststatistic}
	G_{n} = \int_{0}^{\infty} \Lambda^2_n(t) \, w(t) \, \mathrm{d}t,
\end{align}
where
\begin{align*}
	\Lambda_n (t) = \sqrt{n} \left[ \frac{1}{n} \sum\limits_{j = 1}^{n} \left( - \frac{\widehat{k}_n - 1}{Y_{n,j}} + 1 \right) \min \{ Y_{n,j}, t \} - \frac{1}{n} \sum\limits_{j = 1}^{n} \mathds{1} \{ Y_{n,j} \leq t \} \right],
\end{align*}
for $t > 0$, and $w : (0, \infty) \to (0, \infty)$ is a continuous weight function satisfying
\begin{equation} \label{prerequ weight fct}
	\int_0^{\infty} (t^2 + 1) \, w(t) \, \mathrm{d}t < \infty .
\end{equation}
A test based on $G_n$ rejects $H_0$ for large values of the statistic. \\

Special focus will be given to the weight function $w_a(t) = e^{- at}$, where $a > 0$ is some tuning parameter. The appealing feature of this weight function is that $G_{n, a}=\int_{0}^{\infty} \Lambda^2_n(t) \, w_a(t) \, \mathrm{d}t,$ has a closed form expression. Namely, putting $B_{\widehat{k}_n}\left( Y_{j:n} \right) = - \frac{\widehat{k}_n - 1}{Y_{j:n}} + 1$ , where $Y_{1:n} \leq \dotsc \leq Y_{n:n}$ are the order statistics of $Y_{n,1},\ldots,Y_{n,n}$, we obtain
\begin{align} \label{computational form}
	G_{n, a}
	= \, & \frac{2}{n} \sum\limits_{1\leq j < \ell \leq n} \left\{ \frac{e^{- a Y_{\ell:n}}}{a} \big( Y_{j:n} - \widehat{k}_n \big) \left[ - \frac{1}{a} \, B_{\widehat{k}_n}\left( Y_{\ell:n} \right) - 1 \right]\right. \nonumber \\
	& ~~~~~~+ \frac{2}{a^3} \, B_{\widehat{k}_n}\left( Y_{j:n} \right) \, B_{\widehat{k}_n}\left( Y_{\ell:n} \right) \nonumber \\
	& ~~~~~~+ \left. \frac{e^{- a Y_{j:n}}}{a} B_{\widehat{k}_n}\left( Y_{\ell:n} \right) \left[ \frac{1}{a} \, (\widehat{k}_n - 2 - Y_{j:n}) - \frac{2}{a^2} \, B_{\widehat{k}_n}\left( Y_{j:n} \right) - Y_{j:n} \right] \right\} \nonumber \\
	& +  \frac{1}{n} \sum\limits_{j=1}^{n} \left\{ \frac{e^{- a Y_{j:n}}}{a} \left[ 2 \widehat{k}_n - 1 - 2 Y_{j:n} + \big(B_{\widehat{k}_n}\left( Y_{j:n} \right)\big)^2 \, \left( - \frac{2}{a} \, Y_{j:n} - \frac{2}{a^2} \right) \right] \right. \nonumber \\
	& ~~~~~~+ \left. \frac{2}{a^3} \, \big(B_{\widehat{k}_n}\left( Y_{j:n} \right)\big)^2 \right\}
\end{align}
by plain integral calculations.

It is interesting to see that, when restricted to the class of exponential distributions
\begin{align*}
	\mathbf{\mathcal{E}} = \big\{ \mathrm{Exp}(\lambda) = \Gamma(1, \lambda) \, \big| \, \lambda > 0 \big\} ,
\end{align*}
our statistic $G_{n,a}$ reduces to the one proposed in \cite{BH:2000} (see also \cite{BH:2008}). Arguing that $X$ has an exponential distribution if, and only if, the mean residual life function is constant, that is
\begin{equation*}
	\E \big[ X - t \, \big| \, X > t \big]
	= \E X, \quad \text{for each } t > 0,
\end{equation*}
which is equivalent to
\begin{equation} \label{mean res life}
	\E \big[ \min \{ X, t \} \big]
	= F(t) \, \E X, \quad \text{for each } t > 0,
\end{equation}
the authors of \cite{BH:2000} proposed the statistic
\begin{equation*}
	n \int_0^{\infty} \left( \frac{1}{n} \sum\limits_{k=1}^{n} \min \{ U_k, t \} - \frac{1}{n} \sum\limits_{k=1}^{n} \mathds{1}\{ U_k \leq t \} \right)^2 e^{-a t} \, \mathrm{d}t, \quad a > 0,
\end{equation*}
with $U_k = X_k / \overline{X}_n$, $k = 1, \dots, n$, and $\overline{X}_n = n^{-1} \sum_{k=1}^{n} X_k$. The resulting test is consistent against general alternatives (cf. \cite{BH:2000}) and has already been included in the extensive simulation study \cite{ASSV:2017} delivering a remarkable performance in terms of the tests power. We also want to emphasize that our Theorem \ref{thm fixed point statement} is a vast generalization of the characterization (\ref{mean res life}). \\

The paper is organized as follows. In the next section, we will apply the central limit theorem in Hilbert spaces to determine the limit distribution of $G_n$ under the hypothesis $H_0$. We will further use these results in Section \ref{contiguous alternatives} to study the behaviour of our statistic under a sequence of contiguous alternatives that converge to a fixed Gamma distribution at rate $n^{-1/2}$. In Section \ref{consistency}, we establish the consistency of our test based on $G_n$ against general alternatives. Section \ref{empirical} is devoted to the details of the implementation with special focus on the bootstrap technique and the maximum likelihood estimators. Additionally, we compare our new procedure to different classical and modern tests of fit in a finite-sample Monte Carlo power study.

\section{The limit null distribution}
\label{limit null}

As a framework for asymptotic results, we let $L^2 = L^2 \big( (0, \infty), \mathcal{B}_{>0}, w(t) \, \mathrm{d}t \big)$ be the Hilbert space of measurable, square integrable functions $f: (0,\infty) \to \R$ and regard $\Lambda_n$ as a random element of $L^2$. In this setting, $G_n = \norm{\Lambda_n}_{L^2}^2$, where
\begin{equation*}
	\norm{f}_{L^2}^2 = \int_0^{\infty} f^2(t) \, w(t) \, \mathrm{d}t
\end{equation*}
is the $L^2$-norm. We will use the notation $\stackrel{\mathcal{D}}{\longrightarrow}$ for the convergence in distribution of both random variables and random elements of $L^2$. For the remainder of this section, we assume that $H_0$ holds, i.e. $X \sim \Gamma(k,\lambda)$ for some $(k, \lambda) \in (0, \infty)^2$. By the scale invariance of $G_n$, the null distribution of $G_n$ does not depend on the underlying second parameter $\lambda$, so we assume $\lambda = 1$. In view of the bootstrap procedure used to obtain critical values for $G_n$, we consider a triangular array $X_{n,1}, \dots, X_{n,n}$, $n \in \N$, of (rowwise) independent and identically distributed random variables with
\begin{align*}
	X_{n,1} \sim \Gamma(k_n, 1), \quad k_n \in (0, \infty) \quad \text{and} \quad \lim\limits_{n \, \to \, \infty} k_n = k .
\end{align*}
We let $\widehat{k}_n$ be a scale invariant, consistent estimator of $k$, having an expansion
\begin{align} \label{expansion k}
	\sqrt{n} \big( \widehat{k}_n - k_n \big)
	= \frac{1}{\sqrt{n}} \sum\limits_{j = 1}^{n} \Psi_1(X_{n,j}, k_n) + o_{\mathbb{P}}(1),
\end{align}
where $o_{\mathbb{P}}(1)$ denotes convergence to $0$ in probability, as $n \to \infty$. The function $\Psi_1$ is measurable and satisfies
\begin{align} \label{regularity1 k}
	\E [ \Psi_1 (X_{n,1}, k_n) ] = 0, ~~~ \E [ \Psi_1^2(X_{n,1}, k_n) ] < \infty
\end{align}
and
\begin{align} \label{regularity2 k}
	\lim\limits_{n \, \to \, \infty} \E [ \Psi_1^2(X_{n,1}, k_n) ]
	= \E [ \Psi_1^2(X, k) ] .
\end{align}
Note that when implementing the bootstrap, we set $k_n = \widehat{k}_n(\omega)$ for the same element $\omega$ of the underlying probability space that generates the values \linebreak$X_{n,1}(\omega), \ldots, X_{n,n}(\omega)$. Similarly, we assume that $\widehat{\lambda}_n$ is a scale equivariant, consistent estimator of $\lambda$ with
\begin{align} \label{expansion lambda}
	\sqrt{n} \big( \widehat{\lambda}_n^{-1} - 1 \big)
	= \frac{1}{\sqrt{n}} \sum\limits_{j = 1}^{n} \Psi_2(X_{n,j}, k_n) + o_{\mathbb{P}}(1)
\end{align}
and require (\ref{regularity1 k}) and (\ref{regularity2 k}) to hold for $\Psi_2$ also. Under these assumptions, we obtain the following asymptotic result.
\begin{theorem} \label{thm H_0 distr}
	For the test statistic defined in (\ref{teststatistic}) we have
	\begin{align*}
		G_n = \norm{\Lambda_n}_{L^2}^2 \stackrel{\mathcal{D}}{\longrightarrow} \norm{\mathcal{W}}_{L^2}^2 \quad \text{as } n \to \infty.
	\end{align*}
	Here, $\mathcal{W}$ is a centred Gaussian element of $L^2$ with covariance kernel
	\begin{align*}
		\mathcal{K}_k (s,t) = & ~ k(1 + k) P(s, k+2) - 2 k P(s, k+1) + k^2 P(s, k) \\
		& + st \, \E \big[ \big( (k - 1)^2 X^{-2} - (k - 1) X^{-1} \big) \, \mathds{1} \{ X > t \} \big] \\
		& + s k \big( P(t, k+1) - P(s, k+1) \big) + s(1 - k) \big( P(t,k) - P(s,k) \big) \\
		& + s k \big( p(t,k) - p(s,k) \big) +  st \, p(t,k) \\
		& + \E \big[ R_k^X (s) \, r_k^X (t) \big] + \E \big[ R_k^X (t) \, r_k^X (s) \big] + \E \big[ r_k^X (s) \, r_k^X (t) \big],
	\end{align*}
	where $0 < s \leq t < \infty$,
	\begin{align*}
		R_k^X (t) &= (X - k) \, \mathds{1}\{ X \leq t \} + t \, \big( -(k - 1) X^{-1} + 1 \big) \, \mathds{1}\{ X > t \}, \\
		r_k^X (t) &= \Psi_1(X, k) \Big( P(t,k) + t \, \E\big[ X^{-1} \, \mathds{1}\{ X > t \} \big] \Big) \\
		& ~~~- \Psi_2(X, k) \Big( k P(t, k+1) + t \big( 1 - P(t,k) \big) \Big).
	\end{align*}
	Moreover,  $p(\cdot, k) = p(\cdot, k, 1)$ and $P(\cdot, k)$ denote the density and distribution function of the $\Gamma(k,1)$-law, respectively.
\end{theorem}
\begin{proof}
	We define
	\begin{align} \label{V_n element L^2 for change of var}
		V_n (s) = \frac{1}{\sqrt{n}} \sum\limits_{j = 1}^{n} \big( - (\widehat{k}_n - 1) X_{n,j}^{-1} + \widehat{\lambda}_n^{-1} \big) \min \{ X_{n,j}, s \} - \mathds{1} \{ X_{n,j} \leq s \},
	\end{align}
	and observe that a simple change of variable yields
	\begin{align} \label{representation G}
		G_n
		= \widehat{\lambda}_n^{-1} \int_0^{\infty} V^2_n (s) \, w\big(\widehat{\lambda}_n^{-1} s \big) \, \mathrm{d}s
		= \norm{V_n}_{L^2}^2 + o_{\mathbb{P}}(1) .
	\end{align}
	Here, the second equality holds since $\widehat{\lambda}_n \to 1$ in probability, as $n \to \infty$, and
	\begin{align} \label{evade parameter in weight fct}
		\left| \int_0^{\infty} V^2_n (s) \, w(s) \, \mathrm{d}s - \int_0^{\infty} V^2_n (s) \, w\big(\widehat{\lambda}_n^{-1} s \big) \, \mathrm{d}s \right|
		&\leq \sup\limits_{s \, > \, 0} \left| 1 - \frac{w\big( \widehat{\lambda}_n^{-1} s \big)}{w(s)} \right| \norm{V_n}_{L^2}^2 \nonumber \\
		&\stackrel{\mathbb{P}}{\longrightarrow} 0,
	\end{align}
	provided that $V_n \in L^2$ is bounded in probability. Indeed, we will show that $V_n$ converges in distribution to the Gaussian element stated in the theorem. To this end, we notice that
	\begin{align*}
		\widetilde{V}_n (s)
		= & ~ \frac{1}{\sqrt{n}} \sum\limits_{j=1}^{n} \big( -(k_n - 1) X_{n,j}^{-1} + 1\big) \min \{ X_{n,j}, s \} - \mathds{1} \{ X_{n,j} \leq s \} \\
		& ~ + \sqrt{n} (k_n - \widehat{k}_n) \, \E \big[ X^{-1} \min \{ X, s \} \big]
		+ \sqrt{n} (\widehat{\lambda}_n^{-1} - 1) \, \E \big[ \min \{ X, s \} \big]
	\end{align*}
	satisfies $\lVert V_n - \widetilde{V}_n \rVert_{L^2}^2 = o_{\mathbb{P}}(1),$ since both $\sqrt{n} (k_n - \widehat{k}_n)$ and $\sqrt{n} (\widehat{\lambda}_n^{-1} - 1)$ are bounded in probability and, using (\ref{prerequ weight fct}) together with Fubini's theorem and $X_{n,1} \stackrel{\mathcal{D}}{\longrightarrow} X$, we have
	\begin{equation*}
		\norm{ \frac{1}{n} \sum_{j=1}^{n} X_{n,j}^{-1} \min \{ X_{n,j}, \cdot \} - \E \big[ X^{-1} \min \{ X, \cdot \} \big] }_{L^2}^2 = o_{\mathbb{P}}(1)
	\end{equation*}
	and
	\begin{equation*}
		\norm{\frac{1}{n} \sum_{j=1}^{n} \min \{ X_{n,j}, \cdot \} - \E \big[ \min \{ X, \cdot \} \big]}_{L^2}^2
		= o_{\mathbb{P}}(1) .
	\end{equation*}
	From (\ref{expansion k}) and (\ref{expansion lambda}), we obtain
	\begin{align*}
		\left\lVert \widetilde{V}_n - \frac{1}{\sqrt{n}} \sum\limits_{j=1}^{n} W_{n,j} \right\rVert_{L^2}^2 = o_{\mathbb{P}}(1),
	\end{align*}
	where
	\begin{align*}
		W_{n,j} (s)
		&= \left( - (k_n - 1) X_{n,j}^{-1} + 1 \right) \min \{ X_{n,j}, s \} - \mathds{1} \{ X_{n,j} \leq s \} \\
		& ~~~~+ \Psi_1(X_{n,j}, k_n) \, \E \big[ X^{-1} \min \{ X, s \} \big] + \Psi_2(X_{n,j}, k_n) \, \E \big[ \min \{ X, s \} \big],
	\end{align*}
	$j = 1, \dots, n$, are i.i.d. random elements of $L^2$. Since
	\begin{equation*}
		\E \big[ \big( -(k_n - 1) X_{n,1}^{-1} + 1\big) \min \{ X_{n,1}, s \} \big] = \mathbb{P} \big( X_{n,1} \leq s \big)
	\end{equation*}
	by Theorem \ref{thm fixed point statement} and $\E [\Psi_1(X_{n,1}, k_n)] = \E [\Psi_2(X_{n,1}, k_n)] = 0$ by assumption, we have $\E W_{n,1} = 0$. Consequently, by the central limit theorem for Hilbert spaces as stated, for instance, in \cite{CW:1998}, there exists a centred Gaussian random element $\mathcal{W}$ of $L^2$ with $V_n \stackrel{\mathcal{D}}{\longrightarrow} \mathcal{W}$. Since the assumptions on $\Psi_1$ and $\Psi_2$ entail
	\begin{equation*}
		\mathcal{K}_k (s,t) = \lim_{n \, \to \, \infty} \E [W_{n,1}(s) \, W_{n,1}(t)], \quad  s \leq t,
	\end{equation*}
	(\ref{representation G}) and (\ref{evade parameter in weight fct}), combined with Slutsky's Lemma and the continuous mapping theorem, imply the claim.
\end{proof}

\begin{remark}
	The distribution of $\norm{\mathcal{W}}_{L^2}^2$ is known to be that of $\sum_{j=1}^{\infty} \kappa_j N_j^2$, where $N_1, N_2, \dotso$ are independent, standard Gaussian random variables and $\kappa_1, \kappa_2, \dotso$ are the non-zero eigenvalues of the integral operator
	\begin{equation*}
	L^2 \longrightarrow L^2, \quad f \longmapsto \int_0^{\infty} \mathcal{K}_k (\cdot, t) \, f(t) \, w(t) \, \mathrm{d}t .
	\end{equation*}
	We can neither hope to calculate the eigenvalues of this operator explicitly nor do we know the limiting parameter $k$ in practice. Thus, we cannot use asymptotic critical values to implement our test but are forced to operate a parametric bootstrap procedure. With a proof similar to that presented in \cite{Hen:1996}, Theorem \ref{thm H_0 distr} guarantees that in this endeavour, a given level of significance is attained in the limit, as the sample size and the number of bootstrap replications go to infinity.
\end{remark}

\begin{remark}	
	To justify the conditions on the estimators, we will specify them for the maximum likelihood and moments estimators. Since, in the former case, the estimators solve the log-likelihood equations for $k$ and $\lambda$, we have
	\begin{equation} \label{likelihood equation k}
		\log(\widehat{k}_n) - \psi(\widehat{k}_n) = \log(\overline{X}_n) - \frac{1}{n} \sum\limits_{j=1}^{n} \log(X_{n,j}),
	\end{equation}
	with the Digamma function $\psi(k_n) = \Gamma^{\prime}(k_n) / \Gamma(k_n)$, and
	\begin{equation} \label{likelihood equation lambda}
		\widehat{\lambda}_n = \frac{\overline{X}_n}{\widehat{k}_n},
	\end{equation}
	where $\overline{X}_n = n^{-1} \sum_{j = 1}^{n} X_{n,j}$ denotes the sample mean. Noticing that (\ref{likelihood equation k}) may be written as
	\begin{equation*}
		\log(\widehat{k}_n) - \psi(\widehat{k}_n) - \log (k_n) + \psi(k_n)
		= \log(\overline{X}_n) - \log(k_n) + \psi(k_n) - \frac{1}{n} \sum\limits_{j=1}^{n} \log(X_{n,j}) ,
	\end{equation*}
	Taylor expansions of the left-hand side and the first two logarithmic terms on the right-hand side yield
	\begin{align*}
		\left( \frac{1}{k_n} - \psi^{\prime}(k_n) \right) \sqrt{n} \big( \widehat{k}_n - k_n \big)
		&= \sqrt{n} \left( \psi(k_n) - \frac{1}{n} \sum\limits_{j=1}^{n} \log(X_{n,j}) \right) \\
		& ~~~+ \frac{1}{k_n} \sqrt{n} \big( \overline{X}_n - k_n \big) + o_{\mathbb{P}}(1) .
	\end{align*}
	Thus, (\ref{expansion k}) holds with
	\begin{equation*}
		\Psi_1(X_{n,j}, k_n) = \big( 1 - k \psi^{\prime}(k) \big)^{-1} \Big\{ X_{n,j} - k_n + k \big( \psi(k_n) - \log(X_{n,j}) \big) \Big\} .
	\end{equation*}
	Next, we rewrite (\ref{likelihood equation k}) as
	\begin{equation*}
		\log \big(\widehat{k}_n\big) - \psi\big(\widehat{k}_n\big) - \log \big( \overline{X}_n \big) + \psi\big( \overline{X}_n \big)
		= \psi\big( \overline{X}_n \big) - \frac{1}{n} \sum\limits_{j=1}^{n} \log \big(X_{n,j}\big) ,
	\end{equation*}
	and, with another pair of Taylor expansions, we arrive at
	\begin{align*}
		\left( \frac{1}{\overline{X}_n} - \psi^{\prime}\big( \overline{X}_n \big) \right) \sqrt{n} \big( \widehat{k}_n - \overline{X}_n \big)
		&= \sqrt{n} \left( \psi(k_n) - \frac{1}{n} \sum\limits_{j = 1}^{n} \log(X_{n,j}) \right) \\
		& ~~~+ \sqrt{n} \Big( \psi^{\prime}(k_n) \big( \overline{X}_n - k_n \big) \Big) + o_{\mathbb{P}}(1) .
	\end{align*}
	Finally, (\ref{likelihood equation lambda}) provides
	\begin{equation*}
		\sqrt{n} \big( \widehat{\lambda}_n^{-1} - 1 \big)
		= \overline{X}_n^{-1} \big( \widehat{k}_n - \overline{X}_n \big)
	\end{equation*}
	and, using that $\overline{X}_n \to k$ $\mathbb{P}$-a.s., (\ref{expansion lambda}) holds with
	\begin{equation*}
		\Psi_2 (X_{n,j}, k_n) = \big( 1 - k \psi^{\prime}(k) \big)^{-1} \Big\{ \psi^{\prime}(k) (X_{n,j} - k_n) + \psi(k_n) - \log(X_{n,j}) \Big\} .
	\end{equation*}
	Now, we can further specify the covariance operator figuring in Theorem \ref{thm H_0 distr}. For the maximum likelihood estimators we get
	\begin{align*}
		\E \big[ r_k^X(s) \, r_k^X(t) \big]
		&= \big( 1 - k \psi^{\prime}(k) \big)^{-1} \Big( -k e_1(s) e_1(t) - \psi^{\prime}(k) e_2(s)e_2(t) \\
		& ~~~~~~~~~~~~~~~~~\,~~~~~~~~+ e_1(s)e_2(t) + e_1(t)e_2(s) \Big) ,
	\end{align*}
	where we set
	\begin{equation*}
		e_1(t) = P(t,k) + t \, \E\big[ X^{-1} \, \mathds{1}\{ X > t \} \big] \quad \text{and} \quad
		e_2(t) = k P(t, k+1) + t \big( 1 - P(t,k) \big) .
	\end{equation*}
	Moreover, we have
	\begin{align*}
		&\big( 1 - k \psi^{\prime}(k) \big) \, \E \big[ R_k^X(s) \, r_k^X(t) \big] \\
		& ~~~~= \big( e_1(t) - \psi^{\prime}(k) e_2(t) \big) \Big\{ k(1 + k) P(s, k+2) - 2 k P(s, k+1) + k^2 P(s, k) \Big\} \\
		& ~~~~~~+ \big( k e_1(t) - e_2(t) \big) \Big\{ k \psi(k) \big( P(s, k+1) - P(s,k) \big) \\
		& ~~~~~~~~~~~~~~~~~~~~~~~~~~~~~~~~- \E\big[ \log(X) (X - k) \, \mathds{1}\{ X \leq s \} \big] \Big\} \\
		& ~~~~~~+ s \big( e_1(t) - \psi^{\prime}(k) e_2(t) \big) \Big\{ 1 - k P(s, k+1) + (k - 1) P(s,k) - k p(s,k) \Big\} \\
		& ~~~~~~+ s \big( k e_1(t) - e_2(t) \big) \Big\{ \psi(k) p(s,k) \\
		& ~~~~~~~~~~~~~~~~~~~~~~~~~~~~~~~~~- \E\big[ \big( - (k - 1)X^{-1} + 1 \big) \log(X) \, \mathds{1}\{ X > s \} \big] \Big\} .
	\end{align*}
	If moment estimation is employed, the estimators take the form
	\begin{equation*}
		\widehat{k}_n = \frac{(\overline{X}_n)^2}{S_n^2} \quad \text{and} \quad
		\widehat{\lambda}_n = \frac{S_n^2}{\overline{X}_n} ,
	\end{equation*}
	where $S_n^2 = n^{-1} \sum_{j=1}^{n} (X_{n,j} - \overline{X}_n)^2$ is the sample variance. We rewrite the first relation as
	\begin{equation} \label{starting equation moment estima}
		\log(\widehat{k}_n) = 2 \log(\overline{X}_n) - \log(S_n^2)
	\end{equation}
	and realize that Taylor expansions yield
	\begin{align*}
		\sqrt{n} \big(\log(\widehat{k}_n) - \log(k_n) \big)
		&= \sqrt{n} \Big( 2 \log(\overline{X}_n) - 2 \log(k_n) - \big( \log(S_n^2) - \log(k_n) \big) \Big) \\
		&= 2 k_n^{-1} \sqrt{n} (\overline{X}_n - k_n) - k_n^{-1} \sqrt{n} (S_n^2 - k_n) + o_{\mathbb{P}}(1) .
	\end{align*}
	Since
	\begin{equation*}
		\sqrt{n} \big( S_n^2 - k_n \big)
		= \frac{1}{\sqrt{n}} \sum\limits_{j = 1}^{n} \big[ (X_{n,j} - k_n)^2 - k_n \big] + o_{\mathbb{P}}(1) ,
	\end{equation*}
	we verify with yet another Taylor expansion that (\ref{expansion k}) holds with
	\begin{equation*}
		\Psi_1(X_{n,j}, k_n) = 2 X_{n,j} - k_n - (X_{n,j} - k_n)^2 .
	\end{equation*}
	In a similar manner, we derive
	\begin{equation*}
		\sqrt{n} \big( \log(\widehat{k}_n) - \log(\overline{X}_n) \big)
		= k_n^{-1} \, \sqrt{n} (\overline{X}_n - k_n) - k_n^{-1} \, \sqrt{n} (S_n^2 - k_n) + o_{\mathbb{P}}(1)
	\end{equation*}
	from (\ref{starting equation moment estima}) and, via
	\begin{equation*}
		\sqrt{n} \big( \widehat{\lambda}_n^{-1} - 1 \big)
		= \overline{X}_n^{-1} \sqrt{n}\big( \widehat{k}_n - \overline{X}_n \big),
	\end{equation*}
	establish that (\ref{expansion lambda}) is satisfied putting
	\begin{equation*}
		\Psi_2(X_{n,j}, k_n) = \frac{1}{k} \Big\{ X_{n,j} - (X_{n,j} - k_n)^2 \Big\} .
	\end{equation*}
	For the corresponding terms in the covariance operator from Theorem \ref{thm H_0 distr}, we calculate
	\begin{align*}
		\E \big[ r_k^X(s) \, r_k^X(t) \big]
		&= 2k(1 + k) \, e_1(s) e_2(t) - 2(1 + k) \big( e_1(s) e_2(t) + e_1(t) e_2(s) \big) \\
		& ~~~+ k^{-1} (3 + 2k) \, e_2(s) e_2(t)
	\end{align*}
	and
	\begin{align*}
		&\E \big[ R_k^X(s) \, r_k^X(t) \big] \\
		&~~= \left( 2 e_1(t) - k^{-1} e_2(t) \right) \big\{ k(1 + k) P(s, k+2) - k^2 P(s, k+1) \big\} \\
		& ~~~~+ \left( k^{-1} e_2(t) - e_1(t) \right) \big\{ k(k+1)(k+2) P(s, k+3) - 3k^2 (1 + k) P(s, k+2) \\
		& ~~~~~~~~~~~~~~~~~~~~~~~~~~~~~~~~~+ 3k^3 P(s, k + 1) - k^3 P(s,k) \big\} \\
		& ~~~~- k^2 e_1(t) \big( P(s, k+1) - P(s,k) \big) - k s \, e_1(t) p(s,k) \\
		& ~~~~+ \left( k^{-1} e_2(t) - e_1(t) \right) s \big\{ k(1 + k) (1 - P(s, k+2)) + 2k(k - 1) (1 - P(s,k)) \\
		& ~~~~~~~~~~~~~~~~~~~~~~~~~~~~~\,~~~~~+ k(1 - 3k) (1 - P(s, k+1)) + k^2 p(s,k) \big\} \\
		& ~~~~+ \left( 2 e_1(t) - k^{-1} e_2(t) \right) s \big\{ k (1 - P(s, k+1)) + (1 - k) (1 - P(s,k)) \big\} .
	\end{align*}
	Summarizing, we have validated the asymptotic expansions listed in Remark 2.2 of \cite{HME:2012} for both types of estimators.
\end{remark}

\section{The behaviour under contiguous alternatives}
\label{contiguous alternatives}

In contrast to the setting under the null hypothesis, where we had to study the asymptotic properties of our statistic under a triangular array to account for the bootstrap procedure that is run to obtain critical values, we will now look at a triangular array of (rowwise) i.i.d. random variables $Z_{n,1}, \dots, Z_{n,n}$, $n \in \N$, with Lebesgue density
\begin{equation*}
	q_n (x) = p_0(x) \left( 1 + \tfrac{c(x)}{\sqrt{n}} \right), \quad x \in (0, \infty).
\end{equation*}
Here, $p_0$ denotes the density function of a fixed $H_0$-distribution, that is, with the scale invariance of our statistic in mind, the density function of the $\Gamma(k, 1)$-law for some arbitrary, but fixed, $k > 0$. The function $c : (0, \infty) \to \R$ is measurable and bounded, and it satisfies
\begin{align*}
	\int_{0}^{\infty} c(x) \, p_0(x) \, \mathrm{d}x = 0 .
\end{align*}
Additionally, we assume $n$ to be large enough to ensure $q_n \geq 0$. In what follows, we examine
\begin{equation*}
	\mu_n = \bigotimes\limits_{j=1}^{n} (p_0 \, \mathcal{L}^1) \quad \text{and} \quad
	\nu_n = \bigotimes\limits_{j=1}^{n} (q_n \, \mathcal{L}^1),
\end{equation*}
which are measures on $(0, \infty)^n$ endowed with the Borel-$\sigma$-field $\mathcal{B}_{>0}^n$.
\begin{lemma}
	The sequence of measures $\{ \nu_n \}_{n \, \in \, \mathbb{N}}$ is contiguous to $\{ \mu_n \}_{n \, \in \, \mathbb{N}}$.
\end{lemma}
\begin{proof}
	By the absolute continuity of $\nu_n$ with respect to $\mu_n$, the Radon-Nikodym derivative $L_n = \frac{\mathrm{d}\nu_n}{\mathrm{d}\mu_n}$ exists and a Taylor expansion (using the boundedness of $c$) gives
	\begin{align*}
		\log \big( L_n (X_{n,1}, \dots, X_{n,n}) \big)
		= \sum\limits_{j=1}^{n} \left( \frac{c(X_{n,j})}{\sqrt{n}}  - \frac{c^2(X_{n,j})}{2n} \right) + o_{\mathbb{P}}(1),
	\end{align*}
	whenever $(X_{n,1}, \dots, X_{n,n}) \stackrel{\mathcal{D}}{=} \mu_n$. Writing
	\begin{equation*}
		\tau^2 = \int_0^{\infty} c^2(x) \, p_0(x) \, \mathrm{d}x,
	\end{equation*}
	the Lindeberg-Feller central limit theorem and Slutsky's Lemma imply
	\begin{equation*}
		\log \big( L_n \big) \xlongrightarrow{\mathcal{D}}
		\mathcal{N} \left( - \frac{\tau^2}{2}, \tau^2 \right) \quad\mbox{under}\,\mu_n.
	\end{equation*}
	By LeCam's first Lemma (cf. \cite{HSS:1999}, p.253, Corollary 1) we are done.
\end{proof}

\vspace{4mm}
We will now use the statements from Section \ref{limit null} and fundamental results for contiguous measures to derive the (non-degenerate) limit distribution of our statistic under the sequence of contiguous alternatives of the form given above. From this result we conclude that a test based on $G_n$ is able to detect these alternatives. To our knowledge, this setting has not yet been examined in the context of goodness-of-fit tests for Gamma distributions, however, the standard reasoning for this type of contiguous alternatives (see, for instance, \cite{HW:1997} or \cite{BE:2018}) works well.

\begin{theorem}
	Under the triangular array $\big(Z_{n,1}, \dots, Z_{n,n}\big) \stackrel{\mathcal{D}}{=} \nu_n$ we have
	\begin{equation*}
		G_n \stackrel{\mathcal{D}}{\longrightarrow}
		\lVert \mathcal{W} + \zeta \rVert_{L^2}^2 \quad \text{as } n \to \infty,
	\end{equation*}
	where $\zeta(\cdot) = \int_0^{\infty} \eta(x, \cdot) \, c(x) \, p_0(x) \, \mathrm{d}x \in L^2$,
	\begin{align*}
		\eta(x,s) = & \left( - (k - 1) x^{-1} + 1 \right) \min \{ x, s \} - \mathds{1} \{ x \leq s \} \\
		& + \Psi_1(x, k) \, \E \left[ X^{-1} \min \{ X, s \} \right] + \Psi_2(x, k) \, \E \big[ \min \{ X, s \} \big] ,
	\end{align*}
	with $X \sim \Gamma(k,1)$, and $\mathcal{W}$ is the centred Gaussian element of $L^2$ that arose in the proof of Theorem \ref{thm H_0 distr}.
\end{theorem}
\begin{proof}
	When interpreting $V_n$ from (\ref{V_n element L^2 for change of var}) as a random element
	\begin{equation*}
		\big((0, \infty)^n, \, \mathcal{B}_{>0}^n, \, \mu_n \big) \to \big((0, \infty), \, \mathcal{B}_{>0} \big) ,
	\end{equation*}
	we have seen in the proof of Theorem \ref{thm H_0 distr} (cf. (\ref{representation G})) that
	\begin{equation} \label{G_n is L^2 norm of V_n + asymp neglig term}
		G_n = \lVert V_n \rVert_{L^2}^2 + o_{\mu_n}(1) .
	\end{equation}
	Further, we have already established that
	\begin{equation*}
		\lVert V_n - \mathcal{W}_n^* \rVert_{L^2}^2 = o_{\mu_n}(1),
	\end{equation*}
	where $\mathcal{W}_n^*(s) = \tfrac{1}{\sqrt{n}} \sum_{j=1}^{n} W_{n,j}^*(s)$ and
	\begin{align*}
	W_{n,j}^* (s)
	= & ~ \left( - (k - 1) X_{n,j}^{-1} + 1 \right) \min \{ X_{n,j}, s \} - \mathds{1} \{ X_{n,j} \leq s \} \\
	& ~ + \Psi_1(X_{n,j}, k) \, \E \left[ X^{-1} \min \{ X, s \} \right] + \Psi_2(X_{n,j}, k) \, \E \big[ \min \{ X, s \} \big],
	\end{align*}
	with $(X_{n,1}, \dots, X_{n,n}) \stackrel{\mathcal{D}}{=} \mu_n$. Thus, by contiguity,
	\begin{equation} \label{asymp equiv under nu_n}
		\lVert V_n - \mathcal{W}_n^* \rVert_{L^2}^2 = o_{\nu_n}(1) .
	\end{equation}
	Using the boundedness of $c$ and assumption (\ref{regularity1 k}) on $\Psi_1$ and $\Psi_2$, we get, still under $\mu_n$,
	\begin{equation*}
		\lim\limits_{n \, \to \, \infty} \mathrm{Cov} \left( W_{n,1}^*(s), \, c(X_{n,1}) -  \tfrac{1}{2\sqrt{n}} \cdot c^2(X_{n,1}) \right)
		= \zeta(s) .
	\end{equation*}
	Now, for fixed $\ell \in \N$, we let $v \in \R^\ell$ and $s_1, \dots, s_\ell \in (0, \infty)$. Putting $\Sigma = \big( \mathcal{K}_k(s_i, s_j) \big)_{1 \leq i,j \leq \ell}$, where $\mathcal{K}_k$ is the covariance kernel of $\mathcal{W}$ figuring in Theorem \ref{thm H_0 distr}, and letting $\zeta_{\ell} = \big( \zeta(s_1), \dots, \zeta(s_{\ell}) \big)^\top$, the multivariate Lindeberg-Feller central limit theorem implies
	\begin{equation*}
		\left(\begin{array}{c} v_1 \, \mathcal{W}_n^*(s_1) + \dots + v_\ell \, \mathcal{W}_n^*(s_\ell) \\ \log (L_n) \end{array}\right)
		\xlongrightarrow{\mathcal{D}_{\mu_n}}
		\mathcal{N}_{2} \left(
		\left(\begin{array}{c} 0 \\ - \tfrac{\tau^2}{2} \end{array}\right),
		\left(\begin{array}{rr} v^\top \Sigma v \, & \, v^\top \zeta_\ell \\
		\zeta_\ell^\top v & \tau^2 \end{array}\right)
		\right).
	\end{equation*}
	Here, $\mathcal{N}_2$ denotes a bivariate normal distribution determined by its vector of expectations and its covariance matrix. LeCam's third Lemma (see eg. \cite{HSS:1999}, p.259, Lemma 2) yields the convergence of the finite-dimensional distributions of $\mathcal{W}_n^*$ to those of $\mathcal{W} + \zeta$ under $\nu_n$. As $\{\mathcal{W}_n^*\}_{n \in \N}$ is tight under $\mu_n$ and, thus, also under $\nu_n$, we have
	\begin{equation*}
		\mathcal{W}_n^* \xlongrightarrow{\mathcal{D}_{\nu_n}}
		\mathcal{W} + \zeta.
	\end{equation*}
	In combination with (\ref{G_n is L^2 norm of V_n + asymp neglig term}) and (\ref{asymp equiv under nu_n}) (and the contiguity) this finishes the proof.
\end{proof}

\section{Consistency}
\label{consistency}
In this section, we let $X$ be a non-negative, non-degenerate random variable with $\E|X| < \infty$, and $X_1, X_2, \dotso$ are i.i.d. copies of $X$. We assume that the estimators $\widehat{k}_n$, $\widehat{\lambda}_n$ are scale invariant and scale equivariant, respectively, as in Section \ref{Intro}. Furthermore, there exists $(k, \lambda) \in (0,\infty)^2$ so that
\begin{equation}
	(\widehat{k}_n, \widehat{\lambda}_n) \stackrel{\mathbb{P}}{\longrightarrow} (k, \lambda) \quad \text{as } n \to \infty .
\end{equation}
In view of the scale invariance of $G_n$, we assume $\lambda = 1$.
\begin{theorem} \label{thm consistency}
	As $n \to \infty$, we have
	\begin{align*}
		\frac{G_{n}}{n} \stackrel{\mathbb{P}}{\longrightarrow}
		\int_0^{\infty} \Big( \E \left[ \left( - (k - 1) X^{-1} + 1 \right) \min \{ X, t \} \right] - \mathbb{P}\big( X \leq t \big) \Big)^2 w(t) \, \mathrm{d}t .
	\end{align*}
\end{theorem}
\begin{proof}
	By a change of variable, we get
	\begin{align*}
		\frac{G_{n}}{n}
		= \widehat{\lambda}_n^{-1} \int_0^{\infty} V^{* 2}_n(s) \, w\big(\widehat{\lambda}_n^{-1} s\big) \, \mathrm{d}s,
	\end{align*}
	where
	\begin{align*}
		V^*_n(s)
		= \frac{1}{n} \sum\limits_{j=1}^{n} \big( - (\widehat{k}_n - 1) X_j^{-1} + \widehat{\lambda}_n^{-1} \big) \min \{ X_j, s \} - \mathds{1} \{ X_j \leq s \} .
	\end{align*}
	Next, standard Glivenko-Cantelli arguments ensure both
	\begin{equation*}
		\sup\limits_{s \, > \, 0} \left| \frac{1}{n} \sum\limits_{j=1}^{n} X_j^{-1} \min\{ X_j, s \} - \E \big[ X^{-1} \min\{ X, s \} \big] \right|
		\longrightarrow 0 \quad \mathbb{P}\text{-a.s.}
	\end{equation*}
	and
	\begin{equation*}
		\sup\limits_{s \, > \, 0} \left| \frac{1}{n} \sum\limits_{j=1}^{n} \min\{ X_j, s \} - \E \big[ \min\{ X, s \} \big] \right|
		\longrightarrow 0 \quad \mathbb{P}\text{-a.s. }
	\end{equation*}
	as $n\rightarrow\infty$, since all functions involved are continuous and increasing, and the deterministic functions have compact ranges. Thus, an argument similar to (\ref{evade parameter in weight fct}), the assumption $(\widehat{k}_n, \widehat{\lambda}_n) \stackrel{\mathbb{P}}{\longrightarrow} (k, 1)$ and the Glivenko-Cantelli theorem give
	\begin{align*}
		\frac{G_n}{n} \stackrel{\mathbb{P}}{\longrightarrow}
		\int_0^{\infty} \Big( \E \left[ \big( - (k - 1) X^{-1} + 1 \big) \min \{ X, s \} \right] - \mathbb{P} \big( X \leq s \big) \Big)^2 w(s) \, \mathrm{d}s
	\end{align*}
	as $n \to \infty$. Another change of variable yields the claim.
\end{proof}

\vspace{3mm}
Note that Theorem \ref{thm fixed point statement} implies that the limit figuring in Theorem \ref{thm consistency} is positive if $X$ has a non-Gamma distribution. Hence, a test based on $G_n$ is consistent against each alternative with existing expectation and the further assumption discussed in the following remark. From the proof it is evident that, if $(\widehat{k}_n, \widehat{\lambda}_n)$ converges to $(k, \lambda)$ almost surely, we also have almost sure convergence in the theorem. \\ \\
\begin{remark}
	From (\ref{likelihood equation k}) we deduce that, if maximum likelihood estimation is employed, we have $\widehat{k}_n \to k$ $\mathbb{P}$-a.s. where $k$ is the solution of
	\begin{equation} \label{theoretical likelihood eq k}
		\log (k) - \psi (k) = \log (\E X) - \E [\log (X)] .
	\end{equation}
	Under the further assumption $\E |\log (X)| < \infty$ we can assure that
	\begin{equation*}
		0 < \log (\E X) - \E [\log (X)] < \infty
	\end{equation*}
	since $X$ is non-degenerate. We define $h(x) = \log(x) - \psi(x)$, for $x > 0$. From
	\begin{equation*}
		- \frac{1}{2 x} - \frac{1}{x^2}
		< \psi(x) - \log(x)
		< - \frac{1}{2 x},
		\quad  x > 0,
	\end{equation*}
	(cf. 6.3.21 in \cite{AS:1965}), we derive
	\begin{equation*}
		\frac{1}{2 x}
		< h(x)
		< \frac{1}{2 x} + \frac{1}{x^2},
		\quad  x > 0.
	\end{equation*}
	Together with the continuity and monotonicity of $h$ this gives a unique solution $0 < k < \infty$ of (\ref{theoretical likelihood eq k}). Additionally, $\widehat{\lambda}_n = \widehat{k}_n^{-1} \, \overline{X}_n \to \lambda = k^{-1} \, \E X \in (0, \infty)$ which shows that Theorem \ref{thm consistency} is valid for the maximum likelihood estimators. Note that the consistency of our test holds without requiring the existence of the second moment (as e.g. the test in \cite{HME:2012}). \\
	If, instead, moment estimation is employed, we need the condition on the second moment to assure the applicability of Theorem \ref{thm consistency} since
	\begin{equation*}
		\widehat{k}_n \longrightarrow \frac{(\E X)^2}{\mathbb{V}(X)} \quad \text{and} \quad \widehat{\lambda}_n \longrightarrow \frac{\mathbb{V}(X)}{\E X} .
	\end{equation*}
\end{remark}

\section{Monte Carlo Simulation}\label{empirical}
The finite-sample power performance of our test based on $G_{n,a}$ from (\ref{computational form}) is compared to several competitors by means of a Monte Carlo simulation study. All simulations were performed using the statistical computing environment {\tt R}, see \cite{Rstat:2018}. The simulation study was conducted at a $5\%$ nominal level with each entry in Table \ref{Tab.n20} and \ref{Tab.n50} representing the percentage of empirical rejection rates based on $10~000$ replications. Since the limit null distribution depends on the shape parameter $k$, we used the parametric bootstrap procedure proposed in \cite{HME:2012}, section 3, to obtain critical values. In each Monte Carlo run, we computed the approximate maximum likelihood estimator $\widehat{\lambda}_n=\overline{X}_n/\widehat{k}_n$ as in (\ref{likelihood equation lambda}) and
\begin{equation*}
\widehat{k}_n=\left\{\begin{array}{cc} R_n^{-1}\big(0.500876+0.1648852R_n-0.0544274 R_n^2\big), & 0<R_n\le0.5772,\\[1mm]
R_n^{-1}\big(17.79728+11.968477R_n+R_n^2\big)^{-1}\big(8.898919&\\
+9.059950R_n+0.9775373R_n^2\big), & 0.5772< R_n\le 17,\\[1mm]
1/R_n, & R_n>17,\end{array}\right.
\end{equation*}
where $R_n=\log{\overline{X}_n}-\frac1n\sum_{j=1}^n\log X_j$ is the logarithmic ratio of the arithmetic and geometric mean of $X_1,\ldots,X_n$, see \cite{JKB:1994} section 7.2 or \cite{B:2001}. We then generated 500 bootstrap samples from the $\Gamma(\widehat{k}_n,1)$-law and calculated the critical value, as suggested in \cite{GH:2000}, which is given by
\begin{equation*}
	\tilde{c}_B=T^*_{(475)}+0.95\left(T^*_{(476)}-T^*_{(475)}\right).
\end{equation*}
Here, $T^*_{(j)}$ denotes the $j$-th order statistic, $j=1,\ldots,500$, of the bootstrap sample of values of the test statistic $T_1^*,\ldots,T_{500}^*$. This method leads to a more accurate empirical level of the test. In view of the heavy computations due to the bootstrap procedure, we confined the sample sizes to $n=20$ and $n=50$. We didn't consider moment estimators since we expect the same behaviour of the tests as in \cite{HME:2012}, where problems in adhering the significance level are reported.

We compare our new test with classical procedures based on the empirical distribution function, like the Kolmogorov-Smirnov test (KS), Cram\'{e}r-von Mises test (CM), the Anderson-Darling test (AD), and the Watson test (WA), see \cite{delBarrio2000}, section 3, and for a description, see section 3 of \cite{HME:2012}. In addition we implemented the following tests based on the empirical Laplace transform. Writing
\begin{equation*}
\mathcal{Z}_n(t)=\sqrt{n}\left(\frac{\widehat{k}_n}{n}\sum_{j=1}^n\exp(-tY_j)-\frac{1+t}{n}\sum_{j=1}^nY_j\exp(-tY_j)\right),
\end{equation*}
where $Y_j$ is shorthand for $Y_{n,j}=X_j/\overline{X}_n$, the authors of \cite{HME:2012} propose, for a tuning parameter $a>0$, the statistics
\begin{eqnarray*}
T^{(1)}_{n,a} &=&  \int_0^\infty \mathcal{Z}^2_n(t)\exp{(-at)} \: \mbox{d}t\\
&=&\frac
{1}{n} \sum_{j,k=1}^n \left[\frac{Y_j Y_k-\widehat
{k}_n(Y_j+Y_k)+\widehat {k}^2_n}{Y_j+Y_k+a}+
\frac{2Y_jY_k-\widehat {k}_n(Y_j+Y_k)}{(Y_j+Y_k+a)^2}\right.\\&&
~~~+\left.\frac{2Y_jY_k}{(Y_j+Y_k+a)^3}\right]
\end{eqnarray*}
and
\begin{eqnarray*}
T^{(2)}_{n,a} & = &\int_0^\infty \mathcal{Z}^2_n(t)\exp{(-at^2)} \: \mbox{d}t\\
&=&\frac{1}{2n}   \sqrt{\frac{\pi}{a}}
\sum_{j,k=1}^n\left[Y_jY_k+\widehat
{k}^2_n-\widehat{k}_n(Y_j+Y_k)\right]\varphi_{jk}(a)
\\ & & \
+\  \frac{1}{2n}\frac{1}{2a}
\sum_{j,k=1}^n\left[2Y_jY_k-\widehat{k}_n(Y_j+Y_k)\right]
\left[2-\sqrt{\frac{\pi}{a}}(Y_j+Y_k)\varphi_{jk}(a)\right] \\ & &
\ +\  \frac{1}{2n}\frac{1}{4a^2} \sum_{j,k=1}^n
Y_jY_k\left[\left \{\sqrt{\frac{\pi}{a}}(Y_j+Y_k)^2+2\sqrt{\pi a}\right\}
\varphi_{jk}(a)-2(Y_j+Y_k)\right],
\end{eqnarray*}
where
\begin{equation*}
\varphi_{jk}(a)\ = \ \left[1-\Phi  \left( \frac{Y_j+Y_k}{2\sqrt{a}}\right)\right] \exp\!
\left(\frac{(Y_j+Y_k)^2}{4a}\right)
\end{equation*}
and $\Phi(x)=\frac{2}{\sqrt{\pi}}\int_{0}^x\exp\left(-t^2\right)\mbox{dt}$ denotes the error function. To obtain critical values the same parametric bootstrap procedure as referenced above has been used. We chose $T^{(1)}_{n,1}$ and $T^{(2)}_{n,4}$ as representatives for the simulation study, since these choices of the tuning parameter $a$ are recommended by the authors.

For easy comparison to the existing simulation study in \cite{HME:2012}, we considered the same alternative distributions and restate them here (each density being defined for $x > 0$):
\begin{itemize}
	\item[$\bullet$] {the Weibull distribution $W(\vartheta)$ with density $\vartheta x^{\vartheta -1}  \exp(-x^\vartheta)$,}
	\item[$\bullet$] {the inverse Gaussian law $IG(\vartheta)$ with density
	$(\vartheta/2\pi)^{1/2} x^{-3/2}\exp[-\vartheta(x-1)^2/2x]$,}
	\item[$\bullet$] {the lognormal law $LN(\vartheta)$ with density $(\vartheta x)^{-1}  (2\pi)^{-1/2} \exp[-(\log x)^2/(2\vartheta^2)]$,}
	\item[$\bullet$] {the power distribution $PW(\vartheta)$ with density $\vartheta^{-1}   x^{(1-\vartheta)/\vartheta}$, $0 < x \le 1$,}
	\item[$\bullet$] {the shifted-Pareto distribution $SP(\vartheta)$ with density $\vartheta/(1+x)^{1+\vartheta}$,}
	\item[$\bullet$] {the Gompertz law $GO(\vartheta)$ with distribution function
	$1-\exp[\vartheta^{-1}(1-e^x)]$,}
	\item[$\bullet$] {the linear increasing failure rate law $LF(\vartheta)$  with density  $(1+\vartheta x)\exp(-x-\vartheta x^2/2)$.}
\end{itemize}

The best performing tests for each distribution and sample size in Tables \ref{Tab.n20} and \ref{Tab.n50} have been highlighted for easy reference. It is apparent that our tests outperform the classical procedures in most cases (except for the $W(0.5)$, the $PW(4)$ and $SP(1)$ alternatives), although they show a conservative performance under the null hypothesis, i.e. they do not fully exploit the significance level. While $T_{n,1}^{(1)}$ and $T_{n,4}^{(2)}$ have very good power for most alternatives, our new procedures are not too far away or are even better for a reasonable choice of the tuning parameter $a$. When comparing the power of $G_{n,a}$ for the two sample sizes $n=20$ and $n=50$, it seems that a smaller choice of $a$ increases the power, when the tests are among the best procedures, while for some alternatives, e.g. the $LN(1.5)$-law, the performance increases for greater values of $a$.
\begin{table}[t]
\centering
\setlength{\tabcolsep}{.6mm}
\begin{tabular}{rrrrrrrrrrrrrrr}
  \hline
 & $G_{n,0.1}$ & $G_{n,0.25}$ & $G_{n,0.5}$ & $G_{n,0.75}$ & $G_{n,1}$ & $G_{n,1.5}$ & $G_{n,2}$ & $G_{n,3}$ & $T_{n,1}^{(1)}$ & $T_{n,4}^{(2)}$ & $KS$ & $CM$ & $AD$ & $WA$ \\
  \hline
$\Gamma(0.25)$ & 3 & 3 & 3 & 3 & 3 & 3 & 3 & 3 & 4 & 3 & 5 & 5 & 5 & 5 \\
  $\Gamma(0.5)$ & 3 & 3 & 3 & 4 & 4 & 5 & 5 & 5 & 5 & 4 & 5 & 5 & 5 & 5 \\
  $\Gamma(1)$ & 3 & 4 & 5 & 5 & 5 & 5 & 5 & 5 & 5 & 5 & 5 & 5 & 5 & 5 \\
  $\Gamma(5)$ & 5 & 6 & 5 & 3 & 3 & 2 & 2 & 3 & 6 & 5 & 5 & 5 & 5 & 5 \\
  $\Gamma(10)$ & 6 & 5 & 2 & 2 & 2 & 3 & 3 & 3 & 5 & 1 & 5 & 5 & 5 & 5 \\
  \hline $W(0.5)$ & 7 & 8 & 9 & 7 & 9 & 8 & 9 & 8 & {\bf 15} & {\bf 15} & 11 & 12 & 12 & 11 \\
  $W(1.5)$ & 5 & 6 & 7 & {\bf 8} & {\bf 8} & 7 & 7 & 5 & {\bf 8} & {\bf 8} & 6 & 7 & 7 & 6 \\
  $W(3)$ & 15 & 17 & 16 & 13 & 13 & 11 & 12 & 12 & {\bf 18} & 9 & 11 & 12 & 13 & 10 \\
  $IG(0.5)$ & 35 & 37 & 40 & 41 & 43 & 43 & 42 & 40 & {\bf 44} & {\bf 44} & 30 & 37 & 39 & 32 \\
  $IG(1.5)$ & 20 & 21 & 21 & {\bf 22} & 20 & 17 & 13 & 10 & 20 & 19 & 14 & 16 & 17 & 14 \\
  $IG(3)$ & 13 & {\bf 14} & 13 & 10 & 8 & 4 & 3 & 1 & 12 & 9 & 9 & 10 & 11 & 9 \\
  $LN(0.5)$ & {\bf 13} & {\bf 13} & 11 & 8 & 6 & 4 & 2 & 1 & 11 & 8 & 9 & 9 & 10 & 8 \\
  $LN(0.8)$ & 21 & {\bf 22} & {\bf 22} & 21 & 20 & 17 & 15 & 11 & 21 & {\bf 22} & 14 & 16 & 17 & 14 \\
  $LN(1.5)$ & 37 & 38 & 38 & 39 & 39 & 41 & 41 & 42 & {\bf 46} & 43 & 32 & 38 & 39 & 34 \\
  $PW(1)$ & 45 & {\bf 50} & {\bf 50} & 49 & 45 & 40 & 34 & 29 & 44 & 42 & 33 & 42 & 47 & 37 \\
  $PW(2)$ & 22 & 25 & 28 & 30 & {\bf 31} & {\bf 31} & 29 & 24 & 28 & 27 & 20 & 25 & 28 & 22 \\
  $PW(4)$ & 6 & 7 & 8 & 8 & 10 & 10 & 12 & 10 & {\bf 16}  & 10 & 11 & 13 & 15 & 12 \\
  $SP(1)$ & 38 & 40 & 43 & 40 & 45 & 43 & 47 & 45 & {\bf 58} & 56 & 46 & 52 & 52 & 47 \\
  $SP(2)$ & 26 & 26 & 26 & 27 & 26 & 27 & 26 & 25 & 28 & {\bf 30} & 19 & 23 & 23 & 20 \\
  $GO(2)$ & 20 & 25 & 28 & {\bf 29} & 28 & 26 & 24 & 20 & 28 & 28 & 18 & 21 & 23 & 18 \\
  $GO(4)$ & 34 & 39 & {\bf 41} & {\bf 41} & 39 & 35 & 32 & 28 & {\bf 41} & 37 & 25 & 31 & 33 & 26 \\
  $LF(2)$ & 7 & 9 & 12 & 13 & 13 & 13 & 12 & 11 & {\bf 14} & {\bf 14} & 9 & 10 & 11 & 9 \\
  $LF(4)$ & 9 & 12 & 15 & 17 & 16 & 16 & 15 & 13 & {\bf 18} & {\bf 18} & 11 & 12 & 13 & 10 \\
   \hline
\end{tabular}
\caption{Empirical rejection rates for $n=20$, $\alpha=0.05$ ($10~000$ replications)}\label{Tab.n20}
\end{table}

\begin{table}[t]
\centering
\setlength{\tabcolsep}{.6mm}
\begin{tabular}{rrrrrrrrrrrrrrr}
  \hline
 & $G_{n,0.1}$ & $G_{n,0.25}$ & $G_{n,0.5}$ & $G_{n,0.75}$ & $G_{n,1}$ & $G_{n,1.5}$ & $G_{n,2}$ & $G_{n,3}$ & $T_{n,1}^{(1)}$ & $T_{n,4}^{(2)}$ & $KS$ & $CM$ & $AD$ & $WA$ \\
  \hline
  $\Gamma(0.25)$ & 4 & 4 & 4 & 4 & 4 & 4 & 4 & 4 & 4 & 4 & 5 & 4 & 5 & 5 \\
  $\Gamma(0.5)$ & 4 & 4 & 4 & 5 & 5 & 5 & 5 & 5 & 5 & 5 & 5 & 5 & 5 & 5 \\
  $\Gamma(1)$ & 5 & 5 & 5 & 5 & 5 & 6 & 6 & 6 & 5 & 5 & 5 & 5 & 5 & 5 \\
  $\Gamma(5)$ & 5 & 6 & 5 & 4 & 3 & 3 & 3 & 3 & 5 & 6 & 5 & 5 & 5 & 5 \\
  $\Gamma(10)$ & 5 & 5 & 3 & 3 & 3 & 3 & 4 & 4 & 5 & 0 & 5 & 5 & 5 & 5 \\
  \hline $W(0.5)$ & 14 & 17 & 20 & 15 & 22 & 17 & 22 & 17 & 32 & {\bf 33} & 21 & 26 & 26 & 21 \\
  $W(1.5)$ & 9 & 10 & 11 & {\bf 12} & {\bf 12} & 11 & 11 & 9 & {\bf 12} & {\bf 12} & 8 & 9 & 10 & 8 \\
  $W(3)$ & 36 & {\bf 37} & 34 & 32 & 30 & 32 & 33 & 33 & {\bf 37} & 14 & 19 & 23 & 26 & 18 \\
  $IG(0.5)$ & 79 & 81 & 83 & 84 & 85 & 85 & 85 & 83 & {\bf 88} & 84 & 66 & 78 & 82 & 70 \\
  $IG(1.5)$ & 48 & 49 & 50 & 50 & 48 & 44 & 40 & 34 & {\bf 52} & 46 & 30 & 37 & 41 & 30 \\
  $IG(3)$ & {\bf 29} & {\bf 29} & 28 & 25 & 21 & 16 & 14 & 11 & {\bf 29} & 14 & 17 & 20 & 22 & 17 \\
  $LN(0.5)$ & {\bf 27} & 26 & 22 & 19 & 15 & 11 & 10 & 8 & 24 & 8 & 14 & 17 & 18 & 14 \\
  $LN(0.8)$ & {\bf 49} & {\bf 49} & 48 & 47 & 45 & 40 & 36 & 30 & {\bf 49} & 45 & 29 & 36 & 38 & 29 \\
  $LN(1.5)$ & 78 & 79 & 80 & 80 & 81 & 81 & 81 & 79 & {\bf 86} & 82 & 68 & 78 & 80 & 70 \\
  $PW(1)$ & {\bf 95} & 94 & 93 & 90 & 87 & 79 & 72 & 61 & 88 & 86 & 72 & 86 & 92 & 80 \\
  $PW(2)$ & 78 & {\bf 79} & 78 & 78 & 75 & 71 & 65 & 53 & 64 & 71 & 48 & 60 & 69 & 53 \\
  $PW(4)$ & 10 & 12 & 17 & 14 & 20 & 16 & 22 & 16 & 40 & {\bf 41} & 22 & 28 & 34 & 25 \\
  $SP(1)$ & 71 & 74 & 77 & 73 & 80 & 76 & 81 & 77 & {\bf 93} & 92 & 85 & 90 & 90 & 86 \\
  $SP(2)$ & 58 & 58 & 58 & 57 & 57 & 54 & 52 & 48 & {\bf 59} & {\bf 59} & 42 & 49 & 50 & 42 \\
  $GO(2)$ & 61 & 63 & {\bf 64} & 63 & 61 & 55 & 50 & 42 & 63 & {\bf 64} & 42 & 52 & 56 & 42 \\
  $GO(4)$ & 82 & {\bf 83} & 82 & 80 & 77 & 72 & 67 & 62 & 82 & 79 & 58 & 71 & 75 & 61 \\
  $LF(2)$ & 22 & 24 & 27 & 27 & 27 & 26 & 24 & 21 & 30 & {\bf 31} & 18 & 21 & 22 & 16 \\
  $LF(4)$ & 28 & 30 & 33 & 34 & 34 & 31 & 29 & 25 & 37 & {\bf 38} & 21 & 26 & 28 & 20 \\
   \hline
\end{tabular}
\caption{Empirical rejection rates for $n=50$, $\alpha=0.05$ ($10~000$ replications)}\label{Tab.n50}
\end{table}

\section{Conclusions and Outlook}\label{concl}
We established a new characterization of the Gamma distribution by introducing a transformation related to a Stein-type identity. This novel relation can be interpreted as a vast generalization of the characterization of the exponential distribution via the mean residual life function. Based on the explicit representation of our transformation, we proposed a new family of universally consistent goodness-of-fit tests for the Gamma distribution, which can be computed efficiently. We derived the asymptotic distribution under the null hypothesis and under a certain type of contiguous alternatives. Applying a parametric bootstrap technique, we included simulations to show that for an appropriate choice of the tuning parameter, the tests are competitive to existing procedures. It would be beneficial for the application of the test to choose an optimal tuning parameter depending on the data, perhaps using the method suggested in \cite{AS:2015}. An interesting open question is the behaviour of the tests under fixed alternatives: Do we have
\begin{equation*}
\sqrt{n} \left(\frac{G_{n,a}}{n}-\Delta_{k}\right) \stackrel{\mathcal{D}}{\longrightarrow}  \mathcal{N} (0,\sigma_{k}^2) \quad \text{as } n \to \infty,
\end{equation*}
under suitable moment conditions, for some $\sigma_{k}^2>0$. Here,
\begin{equation*}
\Delta_{k} = \int_0^{\infty} \Big( \E \left[ \left( - (k - 1) X^{-1} + 1 \right) \min \{ X, t \} \right] - \mathbb{P}\big( X \leq t \big) \Big)^2 w(t) \, \mathrm{d}t
\end{equation*}
is the constant arising as the limit in Theorem \ref{thm consistency}. This result would open ground to equivalence tests and related theory proposed in \cite{BEH:2017,BH:2017}.

\section{Acknowledgements}
The authors thank Norbert Henze for fruitful discussions and ingenious comments.

\begin{appendix}
\section{Proof of Theorem \ref{thm SC}}\label{apx}
Note that if $X \sim \Gamma(k, \lambda)$, $X$ has Lebesgue density $p = p(\cdot, k, \lambda)$, and the fundamental theorem of calculus implies
\begin{align*}
	\E \left[ f^{\prime}(X) + \left( \frac{k-1}{X} - \frac{1}{\lambda} \right) f(X) \right]
	&= \E \left[ f^{\prime}(X) + \frac{p^{\prime}(X)}{p(X)} \, f(X) \right] \\
	&= \int_0^{\infty} \frac{\mathrm{d}}{\mathrm{d}x} \Big( f(x) p(x) \Big) \, \mathrm{d}x \\
	&= \lim_{x \, \to \, \infty} f(x) p(x) - \lim_{x \, \searrow \, 0} f(x) p(x) \\
	&= 0
\end{align*}
for any $f \in \mathcal{F}$. For the converse, define $f_t : (0, \infty) \to \R$ by
\begin{equation*}
	f_t(x) = \frac{1}{p(x)} \int_0^{x} \Big( \mathds{1}_{(0, t]} (s) - P(t) \Big) p(s) \, 	\mathrm{d}s,\quad  t>0,
\end{equation*}
where $P(t) = \int_0^t p(s) \, \mathrm{d}s$ is the distribution function of the $\Gamma(k, \lambda)$-law. Apparently, $f_t$ is differentiable with
\begin{align} \label{deriv}
	f_t^{\prime}(x)
	&= - \frac{p^{\prime}(x)}{p^2(x)} \int_0^{x} \Big( \mathds{1}_{(0, t]} (s) - P(t) \Big) p(s) \, \mathrm{d}s + \frac{1}{p(x)} \Big( \mathds{1}_{(0, t]} (x) - P(t) \Big) p(x) \nonumber \\
	&= - \left( \frac{k - 1}{x} + \frac{1}{\lambda} \right) f_t(x) + \mathds{1}_{(0, t]} (x) - P(t)
\end{align}
and, since $\left( (k-1)x^{-1} + \lambda^{-1} \right) f_t(x)$ is continuous for $x>0$ and $\mathds{1}_{(0, t]} - P(t)$ is bounded, $f_t^{\prime}$ is locally integrable. Further,
\begin{equation*}
	\lim_{x \, \to \, \infty} f_t(x) p(x) = \int_0^{\infty} \Big( \mathds{1}_{(0, t]} (s) - P(t) \Big) p(s) \, \mathrm{d}s = P(t) - P(t) = 0 .
\end{equation*}
Noting that, for $x < t$, the function takes the form
\begin{equation} \label{form of f_t if x < t}
	f_t(x) = \frac{1}{p(x)} P(x) \big( 1 - P(t) \big) ,
\end{equation}
we infer $\lim_{x \, \searrow \, 0} f_t(x) p(x) = \lim_{x \, \searrow \, 0} P(x) \big( 1 - P(t) \big) = 0$. Arguing that
\begin{equation*}
	0 \leq \frac{1}{p(x)} \, P(x)
	= \int_0^x \left( \frac{s}{x} \right)^{k - 1} e^{\lambda^{-1} (x - s)} \, \mathrm{d}s
	= x e^{\lambda^{-1} x} \int_0^1 z^{k - 1} e^{- \lambda^{-1} x z} \, \mathrm{d}z
	\leq \frac{x e^{\lambda^{-1} x}}{k},
\end{equation*}
equation (\ref{form of f_t if x < t}) also implies $\lim_{x \searrow \, 0} f_t(x) = 0$ and therefore $f_t \in \mathcal{F}$. Thus, the assumption and (\ref{deriv}) yield
\begin{equation*}
	0
	= \E \left[ f_t^{\prime}(X) + \left( \frac{k-1}{X} - \frac{1}{\lambda} \right) f_t(X) \right]
	= \mathbb{P}(X \leq t) - P(t).
\end{equation*}
As $t$ was arbitrary, $X$ follows the $\Gamma(k, \lambda)$-law. \hfill $\square$ \\ \\
\begin{remark}
Since $f_t^{\prime}(x) + \left( (k - 1)^{-1} - \lambda^{-1} \right) f_t(x)$ is bounded for $x > 0$ and the proof depends on $\mathcal{F}$ solely through $f_t$, we can assume that, for any $f \in \mathcal{F}$, the function
\begin{equation*}
	f^{\prime}(x) + \left( (k - 1)x^{-1} - \lambda^{-1} \right) f(x)
\end{equation*}
is integrable with respect to any probability measure and hence the stated expectations exist. We also realize that the requirement $\lim_{x \, \searrow \, 0} f(x) = 0$ for functions in $\mathcal{F}$ was not yet needed. Still, as the last step in the proof of Theorem \ref{thm fixed point statement} relies on this assumption, we had to include it in the characterization given in Theorem \ref{thm SC}.
\end{remark}

\end{appendix}

\bibliography{lit-gamma}
\bibliographystyle{abbrv}      

S. Betsch and B. Ebner, Institute of Stochastics, Karlsruhe Institute of Technology (KIT), Englerstr. 2, D-76133 Karlsruhe:
\\
{\texttt Steffen.Betsch@student.kit.edu} \ \ {\texttt Bruno.Ebner@kit.edu}

\end{document}